\documentclass[12pt,a4paper]{article}
\usepackage{amsmath,amsfonts,amssymb,amsthm,mathtools}
\usepackage{authblk}
\usepackage{latexsym,graphicx}
\usepackage{dsfont}
\usepackage{amscd}

\newcommand{\GR}{G_{\mathrm{R}} } 

\newcommand{\vx}{\mathbf{x}}
\newcommand{\vy}{\mathbf{y}}
\newcommand{\vm}{\mathbf{m}}
\newcommand{\vtx}{\mathbf{\tilde x}}
\newcommand{\vty}{\mathbf{\tilde y}}
\newcommand{\vz}{\mathbf{z}}
\newcommand{\vX}{\mathbf{X}}
\newcommand{\vZ}{\mathbf{Z}}
\newcommand{\vtz}{\mathbf{\tilde z}}
\newcommand{\vw}{\mathbf{w}}
\newcommand{\vtw}{\mathbf{\tilde w}}

\newcommand{\Ref}[1]{(\ref{#1})}

\newcommand{\ee}[1]{{\rm e}^{#1}}
\newcommand{\ii}{{\rm i}}

\newcommand{\R}{{\mathbb R}}
\newcommand{\C}{{\mathbb C}}
\newcommand{\Z}{{\mathbb Z}}
\newcommand{\N}{{\mathbb N}}

\newcommand{\cH}{{\mathcal  H}}
\newcommand{\cE}{{\mathcal  E}}

\newcommand{\cN}{{\mathcal  N}}

\newcommand{\cS}{{\mathcal  S}}
 
\newcommand{\cM}{\mathcal{M}} 
\newcommand{\cA}{\mathcal{A}}

\newcommand{\tx}{\tilde{x}}
\newcommand{\tN}{\tilde{N}} 
\newcommand{\ty}{\tilde{y}}
\newcommand{\tM}{\tilde{M}} 
\newcommand{\tz}{\tilde{z}}

\newtheorem{theorem}{Theorem}[section]
\newtheorem{lemma}[theorem]{Lemma}

\newtheorem{corollary}[theorem]{Corollary}

\theoremstyle{definition}

\theoremstyle{remark}
\newtheorem{remark}[theorem]{Remark}

\title{\Large{\bf{Source identities and kernel functions for deformed (quantum) Ruijsenaars models} }}

\date{\vspace{-0.5cm}\small\today\vspace{-0.5cm}}
\author[1,*]{F. Atai}
\author[2,\dag]{M. Halln{\"a}s}
\author[1,\ddag]{E. Langmann}

\affil[1]{Department of Theoretical Physics\\
KTH Royal Institute of Technology\\
SE-106 91 Stockholm, Sweden}\vspace{2mm} 

\affil[2]{Department of Mathematical Sciences, Loughborough University, Leicestershire LE11 3TU, UK}

\date{\vspace{-1.0cm}\small November 18, 2013}

\begin{document}
\maketitle

\let\oldthefootnote\thefootnote
\renewcommand{\thefootnote}{\fnsymbol{footnote}}
\footnotetext[1]{Electronic address: {\tt farrokh@kth.se}}
\footnotetext[2]{Electronic address: {\tt M.A.Hallnas@lboro.ac.uk}}
\footnotetext[3]{Electronic address: {\tt langmann@kth.se}}
\let\thefootnote\oldthefootnote

\vspace{-1.5cm}

\begin{abstract}
\noindent We consider the relativistic generalization of the quantum $A_{N-1}$ Calogero-Sutherland models due to Ruijsenaars, comprising the rational, hyperbolic, trigonometric and elliptic cases. For each of these cases, we find an exact common eigenfunction for a generalization of Ruijsenaars analytic difference operators that gives, as special cases, many different kernel functions; in particular, we find kernel functions for Chalykh- Feigin-Veselov-Sergeev-type deformations of such difference operators which generalize known kernel functions for the Ruijsenaars models.  We also discuss possible applications of our results. 

\begin{flushleft}{\bf Keywords:} {\scriptsize Exactly solvable models; Ruijsenaars models; Chalykh-Feigin-Veselov-Sergeev type deformation; Kernel functions}\end{flushleft}
\end{abstract}

\section{Motivation and results} 
One special attribute of (quantum) Calogero-Sutherland (CS) systems \cite{C,Su,OP} is the existence of simple and explicit {\em kernel functions} (a function $F(x,y)$ is called kernel function of a pair of Hamilton operators $H(x)$ and $\tilde H(y)$ if $(H(x)-\tilde H(y)-c)F(x,y)=0$, for some constant $c$). Such kernel functions are useful, for example, to construct eigenfunctions and eigenvalues of CS Hamiltonians \cite{EL0,EL1,HL,eCS,HR1}. For many CS systems a remarkable functional identity is known which can be regarded as source of all kernel functions: such {\em source identity} is typically easy to prove, and it implies all kernel function identities for the associated class of CS systems as special cases \cite{HL,sourceAN,sourceBCN} (the first example of a source identity is due to Sen \cite{Sen}, to our knowledge). In particular, source identities also provide kernel functions for certain deformations of the CS Hamiltonian \cite{CFV,Sergeev,SV} which, from a mathematical point of view, provide a natural generalization of the CS model. 

In this paper we present source identities for the relativistic generalizations of the $A_{N-1}$ CS systems due to Ruijsenaars \cite{R}, comprising the rational, trigonometric, hyperbolic and elliptic cases. This not only allows us to recover known kernel functions for these system \cite{RK,HR,KNS,SV2} but also to find some new ones (we discuss the relation of our results to previous ones in more detail in Section~\ref{sec4}); in particular, we obtain kernel functions for a deformation of Ruijsenaar's elliptic operator, which, as far as we know, has not been previously considered in the literature. We believe that our method also is of interest due to its simplicity, and since it unifies and extends various results in the literature. Moreover, it has heuristic value: For example, while various kernel functions for $BC_N$- and Toda-like relativistic systems are already known \cite{HR,KNS}, our results suggest that the lists of known such kernel functions are incomplete, and a systematic method to complete these lists.

To introduce our notation we recall the definition of the Ruijsenaars model: it is (formally) defined by analytic difference operators $S^+_N(\vx;g,\beta)$ and $S^-_N(\vx;g,\beta) $ of the form (our normalization is to simplify formulas further below)
\begin{equation} 
\begin{split} 
\label{Spm} 
S^\pm_N(\vx;g,\beta) = &\frac{s(\ii g\beta)}{\ii g\beta s'(0)}\sum_{j=1}^N \left( \prod_{k\neq j} \left(\frac{s(x_j-x_k\mp\ii g \beta)}{s(x_j-x_k)} \right)^{1/2} \right)
\\ & \times \exp\left(\mp\ii\beta\frac{\partial}{\partial x_j}\right)
\left( \prod_{k\neq j}  \left( \frac{s(x_j-x_k\pm\ii g \beta)}{s(x_j-x_k)} \right)^{1/2}  \right) 
\end{split} 
\end{equation} 
with  $g>0$ the coupling parameter and $\beta>0$ the relativistic deformation parameter (i.e.\ the non-relativistic limit corresponds to $\beta\to 0$); the different cases correspond to 
\begin{equation} 
\label{s} 
s(x) = \begin{cases}  x & \text{ (rational case I)} \\ (1/r)\sin(rx) & \text{ (trigonometric case II)} \\  (a/\pi)\sinh(\pi x/a) & \text{ (hyperbolic case III)} \\ \sigma(x;\pi/2r,\ii a/2)\exp(-\eta_1 r x^2/\pi) & \text{ (elliptic case IV)} \end{cases} 
\end{equation} 
with further parameters $r>0$, $a>0$ ($\sigma(x;\pi/2r,\ii a/2)$ is the Weierstrass elliptic sigma function with periods $2\omega_1=\pi/r$ and $2\omega_2=\ii a$; we set "Planck's constant" $\hbar=1$ since this parameter can be easily introduced in all our formulas by scaling
\begin{equation} 
(g,\beta)\to (g/\hbar, \hbar\beta) , 
\end{equation} 
but otherwise we use the same parameter symbols as Ruijsenaars \cite{R,RGamma}; our conventions for special functions follow \cite{WW}; throughout the paper we use boldface symbols for vectors, e.g., $\vx$ above is short for $(x_1,\ldots,x_N)$; $s'(x)=ds(x)/dx$). Note that, in the elliptic case, $s(x)$ is proportional to the Jacobi theta function $\vartheta_1(rx ;\exp(-ar))$; see \Ref{sell}. While the elliptic case obviously is most general (the others are limiting cases), it is convenient for us to consider the different cases separately; for example,  results for the elliptic cases are true only under certain restrictions (the {\em balancing conditions} given below), whereas they hold true without such restrictions in the other cases. 

We note two important requirements in our approach: $(i)$ We are only interested in {\em common} kernel functions of two pairs of operators $(S^+_N(\vx),S_M^+(-\vy))$ and $(S^-_N(\vx),S_M^-(-\vy))$ (etc.), (ii) we insist on having uniform arguments and proofs for the different cases. While previous works on relativistic kernel functions in \cite{RK,HR} and \cite{KNS} insist on requirement $(i)$ and $(ii)$, respectively, ours is different in that we insist on both requirements at the same time. We note that requirement $(i)$ is necessary for kernel functions which generalize known non-relativistic ones (as explained in \cite{R}), whereas $(ii)$ is natural from a mathematical point of view. 

We now present our main result. The source identities provide a common eigenfunction for the following generalization of the analytic difference operators given by \Ref{Spm}--\Ref{s},
\begin{equation} 
\label{cSpm}
\begin{split} 
\cS^\pm_{\cN}(\vX;\vm)  = &\sum_{J=1}^{\cN} \frac{s(\ii g \beta m_J)}{\ii g\beta s'(0)}\left( \prod_{K\neq J}f_\mp(X_J-X_K;m_J,m_K) \right)
\\  & \times \exp\left(\mp\ii\frac{\beta}{m_J}\frac{\partial}{\partial x_J}\right)
\left( \prod_{K\neq J}f_\pm(X_J-X_K;m_J,m_K) \right)\\ 
\end{split} 
\end{equation} 
\begin{equation} 
\label{fpm} 
f_\pm(x;m,m')  = \begin{dcases} \left(\frac{s(x\pm \ii g\beta(m+m')/2)}{s(x\pm \ii g\beta(m-m')/2)}\right)^{1/2} & \text{ if $m'=m$ or $m'=-\frac1{gm}$} \\ 1& \text{ if $m'=-m$ or $m'=\frac1{gm}$} \end{dcases}
\end{equation} 
depending on parameters $m_J$ taking values in the set 
\begin{equation} 
\label{Lambda} 
\Lambda = \left\{ m_0, -m_0, -\frac1{gm_0}, \frac1{gm_0}\right\} 
\end{equation} 
for some fixed non-zero parameter $m_0$. While we sometimes set $m_0=1$, we find it useful to introduce this parameter to expose the following symmetries in our construction, 
\begin{equation}
\label{symm}  
m_0\to -m_0,\quad m_0\to \frac{1}{gm_0}
\end{equation} 
(note that $\Lambda$ is invariant under these transformations). The pertinent eigenfunction is given by 
\begin{equation} 
\label{Phi0}
\Phi(\vX;\vm) = \prod_{1\leq J<K\leq \cN} \phi(X_J-X_K;m_J,m_K) 
\end{equation} 
\begin{equation}
\label{phi}  
\phi(x;m,m') = \begin{dcases} \left(\frac{G(x+\ii g\beta m -\ii\beta/2m;\beta/m)G(x+\ii\beta/2m;\beta/m)}{G(x-\ii g\beta m +\ii\beta/2m;\beta/m)G(x-\ii\beta/2m;\beta/m)}\right)^{1/2} & \text{ if $m'=m$}\\ 
\frac{G(x-\ii g\beta m/2;\beta/m)}{G(x+\ii g\beta m/2;\beta/m)} & \text{ if $m'=-m$} \\ 
s(x) & \text{ if $m'=\frac1{gm}$}  \\ 
\frac1{\left(s(x+\ii g \beta m_0/2-\ii\beta/2m_0) s(x-\ii g \beta m_0/2+\ii\beta/2m_0)\right)^{1/2}} & \text{ if $m'=-\frac{1}{gm}$}  
\end{dcases} 
\end{equation} 
with $G(x;\alpha)$ a function of two suitably restricted complex variables $x$ and $\alpha$ satisfying a functional equation of the following form,
\begin{equation} 
\label{G1}
\frac{G(x+\ii\alpha/2;\alpha)}{G(x-\ii\alpha/2;\alpha)} = c\,  s(x),\quad c\in\C\setminus\{0\} 
\end{equation} 
(these restrictions include $\Re(\alpha)\neq 0$ and are discussed in more detail in Appendix~\ref{appA}). We note that the constant $c$ in \Ref{G1}, which can depend on $\alpha$, is not important for our purposes (since it can by changed to another constant $c'$ by multiplying $G(x;\alpha)$ with $(c'/c)^{x/\ii\alpha}$). Note that, while we need $G(x;\alpha)$ only for {\em real} $\alpha\neq 0$, we sometimes allow $\alpha$ to be complex for generality.

\begin{theorem}[Relativistic source identities]\label{Theorem1} 
Let $\cN\in\N$, $\vX=(X_1,\ldots,X_{\cN})$ with $X_J$ complex variables such that $X_J\neq X_K$ $\forall J\neq K$, and $\vm=(m_1,\ldots,m_\cN)\in\Lambda^\cN$. Then
\begin{equation} 
\label{SI} 
\left( \cS^\pm_{\cN}(\vX;\vm)- \frac{s(\ii g \beta \sum_{J=1}^\cN m_J)}{\ii g\beta s'(0)}\right) \Phi(\vX;\vm)=0 
\end{equation} 
holds true in the rational, trigonometric, and hyperbolic cases; in the elliptic case, \Ref{SI} holds true if and only if $\vm$ satisfies balancing condition 
\begin{equation} 
\label{BC} 
\sum_{J=1}^\cN m_J =0 . \quad \mathrm{(IV)} 
\end{equation} 
\end{theorem} 

(The proof is given in Section~\ref{sec2}.) 

It is worth noting that, in the corresponding non-relativistic result, the $m_J$ need not be constrained to the set $\Lambda$ but can be arbitrary non-zero complex parameters; see Theorem~\ref{thmnrsource}. Fortunately, this additional constraint in the relativistic case does not affect the applications we have in mind. 

As discussed in more detail in Appendix~\ref{appA}, it is known that solutions to the functional equation in \Ref{G1} exist \cite{RGamma}. For example, in the rational case, 
\begin{equation} 
\label{Grat} 
G(x;\alpha) = \Gamma(1/2 + x/\ii\alpha)\qquad \text{(I)} 
\end{equation} 
with $\Gamma(x)$  the Euler Gamma function, is a solution of \Ref{G1} with $c=1/\ii\alpha$. In this case, 
\begin{equation} 
\label{G2} 
G(x;-\alpha) = G(-x;\alpha). 
\end{equation} 
We find it convenient to impose this condition also in the other cases since then the kernel functions discussed in Section~\ref{sec4} have a somewhat simpler form. Note that \Ref{G2} can always be imposed as follows: take any solution $G(x;\alpha)$ of \Ref{G1} restricted to $\Re(\alpha)>0$ and use \Ref{G2} to extend it to $\Re(\alpha)<0$ (this is explained in more detail in Lemma~\ref{Lemma2}). We stress that the condition in \Ref{G2} is not essential and can be easily dropped; see Remark~\ref{remG1G2}. 

The source identities imply many interesting results as special cases and, in particular, it gives kernel function identities; see Corollaries~\ref{cor1}--\ref{cor5}. The most general such kernel function identity is for pairs of deformed Ruijsenaars operators: such operators $S^\pm_{N,\tN}(\vx,\vtx;g,\beta)$ are equal to the special case $\cN=N+\tN$ ($N,\tN\in\N$) and 
\begin{equation} 
\label{mX}
(m_J,X_J) = \begin{cases} (1,x_J) & \text{ for } J=1,\ldots,N \\ (-1/g,\tx_{J-N}) & \text{ for } J-N=1,\ldots,\tN \end{cases} 
\end{equation} 
of the operators in \Ref{cSpm}, i.e., 
\begin{equation}
\label{SNM} 
\begin{split} 
S^\pm_{N,\tN}(\vx,\vtx;g,\beta) = \sum_{j=1}^N \frac{s(\ii g\beta)}{\ii g\beta s'(0)}A^\mp_j \ee{\mp\ii\beta\frac{\partial}{\partial x_j} }A^\pm_j - \sum_{k=1}^{\tN} \frac{s(\ii \beta)}{\ii g\beta s'(0)}B^\mp_k \ee{\pm\ii g \beta\frac{\partial}{\partial \tx_k} }B^\pm_k  \\ 
A_j^\pm \equiv \left(\prod_{1\leq j'\leq N;j'\neq j} \left(\frac{s(x_j-x_{j'}\pm \ii g\beta)}{s(x_j-x_{j'})}\right)^{1/2} \right) \left( \prod_{k=1}^{\tN}\left(\frac{s(x_j-\tx_k\pm \ii g\beta/2\mp \ii\beta/2)}{s(x_j-\tx_k \pm \ii g\beta/2 \pm \ii\beta/2)}\right)^{1/2} \right)\\ 
B_k^\pm \equiv \left(\prod_{1\leq k'\leq \tN;k'\neq k} \left(\frac{s(\tx_k-\tx_{k'}\mp \ii \beta)}{s(\tx_k-\tx_{k'})}\right)^{1/2} \right) \left( \prod_{j=1}^{N}\left(\frac{s(\tx_k-x_j \mp \ii\beta/2\pm \ii g\beta/2)}{s(\tx_k-x_j \mp \ii\beta/2\mp\ii g\beta/2)}\right)^{1/2} \right), 
\end{split} 
\end{equation} 
and a common kernel functions for two pairs of such operators is given in Corollary~\ref{cor5}. 

The plan of this paper is as follows. Section~\ref{sec2} contains a proof of Theorem~\ref{Theorem1}. In Section~\ref{sec3} we state interesting special cases of Theorem~\ref{Theorem1}, including the kernel function identities. Section~\ref{sec4} contains a comparison with previous results in the literature, in particular, results in the non-relativistic limit (Section~\ref{sec4.1}) and kernel function identities obtained in \cite{RK,HR,KNS,SV2} (Section~\ref{sec4.2}).  We conclude with some remarks in Section~\ref{sec6}. A concise review on Gamma functions, i.e.\ functions satisfying the relations in \Ref{G1}, and a proof of a key lemma, are deferred to Appendices~\ref{appA} and \ref{appB}, respectively.

\section{Proof of main result}
\label{sec2} 
We start by stating three preliminary results needed in our proof of the source identities. 

A key result we need is the following:

\begin{lemma} 
\label{Lemma1} 
Let $\cN\in\N$, $\gamma\in\C$ such that $\Im(\gamma)\neq 0$, $\vZ\in \C^\cN$ such that $Z_J\neq Z_K$ $\forall J\neq K$, and $m_J\in\R\setminus\{0\}$ for $J=1,2,\ldots,\cN$. Then 
\begin{equation} 
\label{WH} 
\sum_{J=1}^\cN s(\gamma m_J) \prod_{K\neq J}\frac{s(Z_J-Z_K+\gamma m_K)}{s(Z_J-Z_K)} = s\left(\gamma\sum_{J=1}^\cN m_J\right) 
\end{equation} 
holds true in the rational, trigonometric, and hyperbolic cases; in the elliptic case, \Ref{WH} holds true if and only if the $m_J$ satisfy the balancing condition in \Ref{BC}.
\end{lemma} 
 
This result is not new (see e.g.\  \cite{KNS}, Lemma~3.1), but for the convenience of the reader we include a proof in Appendix~\ref{appB}.
 
 We need another (simple) result: 
  
 \begin{lemma}
 \label{Lemma2} 
 For $x$ a complex variable and $A$, $\alpha$ two complex parameters with non-zero real parts, the two analytical difference equations 
 \begin{equation} 
 \label{Fpm} 
 \frac{F(x\mp\ii\alpha/2)}{F(x\pm\ii\alpha/2)} = \frac{s(x\mp\ii A)}{s(x\pm \ii A)}
 \end{equation} 
  have a common solution 
 \begin{equation} 
 \label{F1} 
 F(x) = \frac{G(x+\ii A;\alpha)}{G(x-\ii A;\alpha)}
 \end{equation} 
 with $G(x;\alpha)$ any function satisfying the condition in \Ref{G1} for $\Re(\alpha)>0$ and extended to $\Re(\alpha)<0$ by \Ref{G2}. 
 \end{lemma} 
\begin{proof} 
We first consider the lower equation in \Ref{Fpm}, i.e.\ the equation obtained by setting $\mp$ to $+$ and $\pm$ to $-$. If $G(x;\alpha)$ solves \Ref{G1} for all $\alpha$ such that $\Re(\alpha)\neq 0$, then obviously $F_1(x)\equiv G(x+\ii A;\alpha)/G(x-\ii A;\alpha)$ solves the lower equation in \Ref{Fpm}; that $F_1(x)$ also solves the upper equation in \Ref{Fpm} is checked by a simple computation. In a similar manner one checks that $F_2(x)\equiv G(-x-\ii A;-\alpha)/G(-x+\ii A;-\alpha)$ is a common solution of both equations in \Ref{Fpm}. Thus $F(x)$, which is equal to $F_1(x)$ for $\Re(\alpha)>0$ and equal to $F_2(x)$ for $\Re(\alpha)<0$, is also a common solution to both equations.
\end{proof} 

We now are ready to prove the relativistic source identities.

\begin{proof}[Proof of Theorem~\ref{Theorem1}]
We construct a common eigenfunction $\Phi(\vX;\vm)$ of both operators in \Ref{cSpm}. For that we make the ansatz in \Ref{Phi0}, compute
\begin{equation} 
\label{cc} 
\begin{split} 
\cS^\pm_{\cN}(\vX;\vm)\Phi(\vX;\vm) = \sum_{J=1}^{\cN} \frac{s(\ii g\beta m_J)}{\ii g\beta s'(0) } \Biggl(\prod_{K\neq J}f_\mp(X_J-X_K;m_J,m_K)\\ \times f_{\pm}(X_J-X_K\mp\ii\beta/m_J;m_J,m_K)\Biggr) \Biggl(\prod_{K>J}\frac{\phi(X_J-X_K\mp\ii\beta/m_J;m_J,m_K)}{\phi(X_J-X_K;m_J,m_K)}\Biggr)\\ \times  
\Biggl(\prod_{K<J}\frac{\phi(X_K-X_J\pm\ii\beta/m_J;m_K,m_J)}{\phi(X_K-X_J;m_K,m_J)}\Biggr)\Phi(\vX;\vm), 
\end{split} 
\end{equation} 
and determine the functions $\phi(X;m_J,m_K)$ such that this is identical with 
\begin{equation} 
\label{cc1}
\sum_{J=1}^{\cN}  \frac{s(\gamma_\pm m_J)}{\gamma_\pm s'(0) } \left( \prod_{K\neq J}\frac{s(X_J-X_K+\xi_\pm(m_J)-\xi_\pm(m_K) +\gamma_\pm m_K)}{s(X_J-X_K+\xi_\pm(m_J)-\xi_\pm(m_K))}\right)\Phi(\vX;\vm)
\end{equation} 
for some functions $\xi_\pm(m)$ and some constants $\gamma_\pm$ to be found. Then Lemma~\ref{Lemma1} (with $Z_J=X_J+\xi_\pm(m_J)$ and $\gamma=\gamma_\pm$) implies the result.  (Note that $\gamma_\pm$ must be equal to $\ii g\beta$, up to a sign, but the correct choice is not clear at this point.) 

Comparing \Ref{cc} and \Ref{cc1} we obtain the conditions 
\refstepcounter{equation}
\begin{equation} 
\label{c1} 
\begin{split} 
f_\mp(X\pm\ii\beta/2m;m,m')f_{\pm}(X\mp\ii\beta/2m;m,m')\frac{\phi(X\mp\ii\beta/2m;m,m')}{\phi(X\pm\ii\beta/2m;m,m')} \\ = \frac{s(X\pm\ii\beta/2m+\xi_\pm(m)-\xi_\pm(m') +\gamma_\pm m')}{s(X\pm\ii\beta/2m+\xi_\pm(m)-\xi_\pm(m'))} 
\end{split} \tag{\theequation{}a}  
\end{equation} 
for $K>J$ (we renamed $X_J-X_K\mp\ii\beta/2m_J\equiv X$ and $(m_J,m_K)=(m,m')$) and 
\begin{equation} 
\label{c2} 
\begin{split} 
f_\mp(-X\pm\ii\beta/2m';m',m)f_{\pm}(-X\mp\ii\beta/2m';m',m)\frac{\phi(X\pm\ii\beta/2m';m,m')}{\phi(X\mp\ii\beta/2m';m,m')} \\ 
= \frac{s(X\mp\ii\beta/2m'-\xi_\pm(m')+\xi_\pm(m) -\gamma_\pm m)}{s(X\mp\ii\beta/2m'-\xi_\pm(m')+\xi_\pm(m))}  
\end{split} \tag{\theequation{}b} 
\end{equation} 
for $K<J$  (we renamed $X_K-X_J\pm\ii\beta/2m_J\equiv X$ and $(m_J,m_K)\equiv (m',m)$ and used $s(-x)=-s(x)$). We thus get four conditions for each function $\phi(X;m,m')$. 

We now consider these conditions for the different cases $(m,m')=(m,m)$, $(m,-m)$, $(m,1/gm)$, $(m,-1/gm)$, for all $m\in\Lambda$. 

For $(m,m')=(m,m)$ we get from \Ref{c1}
\refstepcounter{equation}
\begin{equation} 
\label{c11} 
\begin{split} 
\frac{\phi(X\mp\ii\beta/2m;m,m)}{\phi(X\pm\ii\beta/2m;m,m)}  = 
\frac{s(X\pm\ii\beta/2m+\gamma_\pm m)}{s(X\pm\ii\beta/2m)}\\ \times \left(\frac{s(X\pm\ii\beta/2m)s(X\mp\ii\beta/2m)}{s(X\pm\ii\beta/2m\mp\ii g\beta m)s(X\mp\ii\beta/2m\pm\ii g\beta m)}\right)^{1/2}
\end{split}  \tag{\theequation{}a} 
\end{equation} 
and from \Ref{c2} 
\begin{equation} 
\label{c21} 
\begin{split} 
\frac{\phi(X\pm\ii\beta/2m;m,m)}{\phi(X\mp\ii\beta/2m;m,m)} 
= \frac{s(X\mp\ii\beta/2m-\gamma_\pm m)}{s(X\mp\ii\beta/2m)}\\ \times \left(\frac{s(X\mp\ii\beta/2m)s(X\pm\ii\beta/2m)}{s(X\mp\ii\beta/2m\pm\ii g\beta m)s(X\pm\ii\beta/2m\mp\ii g\beta m)}\right)^{1/2}
\end{split}  \tag{\theequation{}b} 
\end{equation} 
(we inserted \Ref{fpm} and used $s(-x)=-s(x)$). The choice 
\begin{equation} 
\label{apm} 
\gamma_\pm = \mp \ii g\beta 
\end{equation} 
reduces these four conditions to two, i.e., the upper/lower equation in \Ref{c11} becomes 
\begin{equation} 
\label{phi1}
\begin{split} 
\frac{\phi(X\mp\ii\beta/2m;m,m)}{\phi(X\pm\ii\beta/2m;m,m)}  = \left(\frac{s(X\mp\ii\beta/2m)s(X\mp\ii g\beta m\pm\ii\beta/2m)}{s(X\pm\ii\beta/2m)s(X\pm\ii g\beta m\mp\ii\beta/2m)}\right)^{1/2}
\end{split} 
\end{equation} 
equal to the lower/upper relation in \Ref{c21}. According to Lemma~\ref{Lemma2}, the equations in \Ref{phi1} have the solution $\phi(X;m,m)$ in \Ref{phi}. 

For $(m,m')=(m,-m)$ and with \Ref{fpm} and \Ref{apm} the four conditions in \Ref{c1} and \Ref{c2} reduce to the following two 
\begin{equation} 
\label{c12A} 
\frac{\phi(X\mp\ii\beta/2m;m,-m)}{\phi(X\pm\ii\beta/2m;m,-m)} = \frac{s(X\pm\ii\beta/2m+\xi_\pm(m)-\xi_\pm(-m) \pm \ii g\beta m)}{s(X\pm\ii\beta/2m+\xi_\pm(m)-\xi_\pm(-m))} 
\end{equation} 
(i.e.\ the two pairs of conditions are identical). For 
\begin{equation} 
\label{xi1} 
\xi_\pm(m)-\xi_\pm(-m)= \mp\ii\beta/2m\mp \ii g\beta m/2 \qquad \quad \forall m\in\Lambda 
\end{equation} 
the latter are identical with 
\begin{equation} 
\label{c12B} 
\frac{\phi(X\mp\ii\beta/2m;m,-m)}{\phi(X\pm\ii\beta/2m;m,-m)} = \frac{s(X \pm \ii g\beta m/2)}{s(X\mp \ii g\beta m/2)} . 
\end{equation} 
Thus, according to Lemma~\ref{Lemma2}, we obtain the  solution $\phi(X;m,-m)$ in \Ref{phi} provided that the condition in \Ref{xi1} is satisfied. 

For $(m,m')=(m,1/gm)$ and with \Ref{fpm} and \Ref{apm}, \Ref{c1} and \Ref{c2} become  
\refstepcounter{equation}
\begin{equation} 
\label{c41} 
\begin{split} 
\frac{\phi(X\mp\ii\beta/2m;m,1/gm)}{\phi(X\pm\ii\beta/2m;m,1/gm)} = \frac{s(X\mp\ii\beta/2m+\xi_\pm(m)-\xi_\pm(1/gm))}{s(X\pm\ii\beta/2m+\xi_\pm(m)-\xi_\pm(1/gm))} 
\end{split} \tag{\theequation{}a}  
\end{equation} 
and 
\begin{equation} 
\label{c42} 
\begin{split} 
\frac{\phi(X\pm \ii g\beta m/2;m,1/gm)}{\phi(X\mp \ii g\beta m/2;m,1/gm)}  = \frac{s(X\pm\ii g\beta m/2+\xi_\pm(m)-\xi_\pm(1/gm))}{s(X\mp \ii g\beta m/2+\xi_\pm(m)-\xi_\pm(1/gm))}  
\end{split} \tag{\theequation{}b} , 
\end{equation} 
respectively. The condition 
\begin{equation} 
\label{xi2} 
\xi_\pm(m)-\xi_\pm(1/gm)=  0 \quad \forall m\in\Lambda 
\end{equation} 
simplifies these four equations so as to allow the common solution $\phi(X;m,1/gm)$ in \Ref{phi}. 

For $(m,m')=(m,-1/gm)$ and with \Ref{fpm} and \Ref{apm}, \Ref{c1} and \Ref{c2} become 
\refstepcounter{equation}
\begin{equation} 
\label{c14} 
\begin{split} 
\frac{\phi(X\mp\ii\beta/2m;m,-1/gm)}{\phi(X\pm\ii\beta/2m;m,-1/gm)}  = \frac{s(X\pm 3\ii\beta/2+ \xi_\pm(m)-\xi_\pm(-1/gm))}{s(X\pm\ii \beta /2m+\xi_\pm(m)-\xi_\pm(-1/gm))}\\
\times \left(\frac{s(X\mp\ii g\beta m/2)s(X\pm \ii g\beta m/2)}{s(X\mp\ii g\beta m/2\pm\ii\beta/m)s(X\pm\ii g\beta m/2\mp\ii\beta/m)}\right)^{1/2}
\end{split} \tag{\theequation{}a}  
\end{equation} 
and 
\begin{equation} 
\label{c24} 
\begin{split} 
\frac{\phi(X\mp \ii g\beta m/2;m,-1/gm)}{\phi(X\pm \ii g\beta m/2;m,-1/gm)} 
= \frac{s(X\pm 3\ii g\beta m/2+\xi_\pm(m)-\xi_\pm(-1/gm))}{s(X\pm \ii g\beta m/2+\xi_\pm(m)-\xi_\pm(-1/gm))}\\
\times \left(\frac{s(X\pm\ii\beta/2m)s(X\mp \ii\beta/2m)}{s(X\pm\ii g\beta m\mp\ii\beta/2m)s(X\mp\ii g\beta m\pm\ii\beta/2m)}\right)^{1/2}, 
\end{split} \tag{\theequation{}b} 
\end{equation} 
respectively (we used $s(-x)=-s(x)$). The condition 
\begin{equation} 
\label{xi3} 
\xi_\pm(m)-\xi_\pm(-1/gm)= \mp\ii\beta/2m\mp \ii g\beta m/2 \qquad \quad \forall m\in\Lambda 
\end{equation} 
simplifies these four equations to 
\refstepcounter{equation}
\begin{equation} 
\label{c14b} 
\begin{split} 
\frac{\phi(X\mp\ii\beta/2m;m,-1/gm)}{\phi(X\pm\ii\beta/2m;m,-1/gm)}  = \left(\frac{s(X\mp\ii g\beta m/2\pm\ii\beta/m)s(X\pm \ii g\beta m/2)}{s(X\pm\ii g\beta m/2\mp\ii\beta/m)s(X\mp\ii g\beta m/2)}\right)^{1/2}
\end{split} \tag{\theequation{}a}  
\end{equation} 
and 
\begin{equation} 
\label{c24b} 
\begin{split} 
\frac{\phi(X\mp \ii g\beta m/2;m,-1/gm)}{\phi(X\pm \ii g\beta m/2;m,-1/gm)} = \left(\frac{s(X\pm\ii g\beta m\mp\ii\beta/2m)s(X\pm\ii\beta/2m)}{s(X\mp\ii g\beta m\pm\ii\beta/2m)s(X\mp \ii\beta/2m)}\right)^{1/2} 
\end{split} \tag{\theequation{}b} 
\end{equation} 
which have the common solution $\phi(X;m,-1/g m)$ in \Ref{phi}, for all $m\in\Lambda$. 

To finish the proof we are left to show that there exist functions $\xi_\pm(m)$ satisfying the conditions in \Ref{xi1}, \Ref{xi2} and \Ref{xi3}. We do this by writing down such functions:
\begin{equation} 
\label{xipm} 
\xi_\pm(m) = \mp\frac{\ii g^2\beta}{4}\left( (m_0^2+1/(m_0 g)^2 + 1/g)m-m^3\right)  
\end{equation} 
(it is trivial to check that all conditions are satisfied). 

To summarize: We showed that \Ref{cc} with $\phi(x;m,m')$ in \Ref{phi} is identical with \Ref{cc1} with $\gamma_\pm$ in \Ref{apm} and $\xi_\pm(m)$ in \Ref{xipm}, and thus Lemma~\ref{Lemma1} implies the result. 
\end{proof} 

From our proof above it is easy to also deduce the following result. 

\begin{lemma} 
\label{LemmaA} 
For $\cN\in\N$, $\beta>0$, $g$ real, $\vm\in\Lambda^\cN$ with $\Lambda$ in \Ref{Lambda} and $m_0>0$, $\cS^\pm_{\cN}(\vX;\vm)$ in \Ref{cSpm} and $\Phi(\vX;\vm)$ in \Ref{Phi0} with $s(x)$ in \Ref{s}, the analytical difference operator
\begin{equation} 
\frac{1}{\Phi(\vX;\vm)} \cS^\pm_{\cN}(\vX;\vm)\Phi(\vX;\vm) 
\end{equation} 
is identical with 
\begin{equation} 
\label{MR}
\sum_{J=1}^{\cN}  \frac{s(\ii g\beta m_J)}{\ii g\beta s'(0) } \left( \prod_{K\neq J}\frac{s(X_J-X_K \mp\ii \xi(m_J,m_K) \mp\ii g\beta m_K)}{s(X_J-X_K\mp\ii \xi(m_J,m_K))}\right) \exp\left(\mp\ii \frac{\beta}{m_J}\frac{\partial}{\partial X_J} \right)
\end{equation} 
\begin{equation} 
\label{xi} 
\xi(m,m') =  \begin{dcases} 0 & \text{ if $m'=m$ or $m'=\frac1{gm}$} \\   \beta/2m +  g\beta m/2& \text{ if $m'=-m$ or $m'=-\frac1{gm}$} . \end{dcases}
\end{equation} 
\end{lemma} 

\begin{proof} 
We compute 
\begin{equation} 
\label{cc9} 
\begin{split} 
\frac{1}{\Phi(\vX;\vm)} \cS^\pm_{\cN}(\vX;\vm)\Phi(\vX;\vm) = \sum_{J=1}^{\cN} \frac{s(\ii g\beta m_J)}{\ii g\beta s'(0) } \Biggl(\prod_{K\neq J}f_\mp(X_J-X_K;m_J,m_K)\\ \times f_{\pm}(X_J-X_K\mp\ii\beta/m_J;m_J,m_K)\Biggr) \Biggl(\prod_{K>J}\frac{\phi(X_J-X_K\mp\ii\beta/m_J;m_J,m_K)}{\phi(X_J-X_K;m_J,m_K)}\Biggr)\\ \times  \Biggl(\prod_{K<J}\frac{\phi(X_K-X_J\pm\ii\beta/m_J;m_K,m_J)}{\phi(X_K-X_J;m_K,m_J)}\Biggr) \exp\left( \mp \ii\frac{\beta}{m_J}\frac{\partial}{\partial X_J}\right) 
\end{split}
\end{equation} 
and note that our computation above implies that this is equal to 
\begin{equation} 
\sum_{J=1}^{\cN}  \frac{s(\gamma_\pm m_J)}{\gamma_\pm s'(0) } \left( \prod_{K\neq J}\frac{s(X_J-X_K+\xi_\pm(m_J)-\xi_\pm(m_K) +\gamma_\pm m_K)}{s(X_J-X_K+\xi_\pm(m_J)-\xi_\pm(m_K))}\right) \exp\left( \mp \ii\frac{\beta}{m_J}\frac{\partial}{\partial X_J}\right) . 
\end{equation} 
Insert \Ref{apm}, \Ref{xi1}, \Ref{xi2}, \Ref{xi3} and write $\xi_\pm(m)-\xi_\pm(m')$ as $\mp\ii\xi(m,m')$ to obtain the result. 
\end{proof} 

Note that Lemma~\ref{LemmaA} holds true in general in all cases (i.e.\ there is no balancing condition in the elliptic case). We state this result since it is used to derive some facts discussed in Section~\ref{sec4}. 
 
\section{Special cases} 
\label{sec3} 
We now state various interesting implications of the source identities (see Section~\ref{sec4} for a comparison with known results). For simplicity we set $m_0=1$, without loss of generality. 

Choosing $\cN=N$ ($N\in\N$) and $(m_J,X_J)=(1,x_J)$ for $J=1,\ldots,N$ one obtains an important common eigenfunction of the analytical difference operators $S^\pm_N(\vx;g,\beta)$ in \Ref{Spm}: 

\begin{corollary} 
\label{cor1} 
The function
\begin{equation} 
\label{PsiN}
\Psi_N(\vx;g,\beta) = \prod_{1\leq j<k\leq N}\left(\frac{G(x_j-x_k+\ii g\beta  -\ii\beta/2;\beta)G(x_j-x_k+\ii\beta/2;\beta)}{G(x_j-x_k-\ii g\beta  +\ii\beta/2;\beta)G(x_j-x_k-\ii\beta/2;\beta)}\right)^{1/2}
\end{equation} 
obeys 
\begin{equation} 
\left( S^\pm_N(\vx;g,\beta) - \frac{s(\ii g\beta N)}{\ii g\beta s'(0)) } \right) \Psi_N(\vx;g,\beta) =0 
\end{equation} 
for arbitrary $N$ in the rational, trigonometric and hyperbolic cases. 
\end{corollary} 

Note that this  result is {\em not} true in the elliptic case (since the balancing condition in this case is $N=0$).

Choosing $\cN=N+M$ ($N,M\in\N$) and 
\begin{equation*} 
(m_J,X_J) = \begin{cases} (1,x_J) & \text{ for } J=1,\ldots,N \\ (-1,y_{J-N}) & \text{ for } J-N=1,\ldots,M \end{cases} 
\end{equation*} 
we obtain common kernel functions for the two pairs of analytical difference operators $(S^\pm_N(\vx;g,\beta),S^\pm_M(-\vy;g,\beta))$ (we also use \Ref{G2} for simplicity): 

\begin{corollary} 
\label{cor2} 
The function
\begin{equation} 
\label{FNM}
F_{N,M}(\vx,\vy;g,\beta) = \Psi_N(\vx;g,\beta)\Psi_M(-\vy;g,\beta) \prod_{j=1}^N\prod_{k=1}^M \frac{G(x_j-y_k-\ii g\beta /2;\beta)}{G(x_j-y_k+\ii g\beta /2;\beta)} 
\end{equation} 
obeys 
\begin{equation} 
\label{SI1}
\left( S^\pm_N(\vx;g,\beta) - S^\pm_M(-\vy;g,\beta) - \frac{s(\ii g\beta (N-M))}{\ii g\beta s'(0) } \right)  F_{N,M}(\vx,\vy;g,\beta) =0
\end{equation} 
for arbitrary $N,M$ in the rational, trigonometric and hyperbolic cases, and for 
\begin{equation} 
N-M=0 \quad \mathrm{(IV)} 
\end{equation} 
in the elliptic case. 
\end{corollary} 

\begin{remark} 
\label{remv} 
Translation invariance allows to introduce an arbitrary complex parameter $v$ into the kernel function in \Ref{FNM} as follows: replace $x_J$ by $x_J+v$ in the specializations of the source identities above. This does not change the operators $(S^\pm_N(\vx;g,\beta),S^\pm_M(-\vy;g,\beta))$ and the kernel function identity, but the kernel function becomes $v$-dependent. Note also that, by changing the sign of $\vy$, one can rewrite \Ref{SI1} as a kernel function identity for the operators $(S^\pm_N(\vx;g,\beta),S^\pm_M(\vy;g,\beta))$. Similar remarks applies to the kernel functions in Corollaries~\ref{cor3} and \ref{cor5} below. 
\end{remark} 

\begin{remark}
\label{remG1G2} 
The restriction on the function $G(x;\alpha)$ in \Ref{G2} can be removed: the only change is that $\Psi_M(-\vy;g,\beta)$ in \Ref{FNM} is replaced by 
\begin{equation} 
\label{tPsiM}
\Psi^{(-)}_M(-\vy;g,\beta) = \prod_{1\leq j<k\leq M}\left(\frac{G^{(-)}(y_k-y_j+\ii g\beta  -\ii\beta/2;\beta)G^{(-)}(y_k-y_j+\ii\beta/2;\beta)}{G^{(-)}(y_k-y_j-\ii g\beta  +\ii\beta/2;\beta)G^{(-)}(y_k-y_j-\ii\beta/2;\beta)}\right)^{1/2} 
\end{equation} 
with $G^{(-)}(x;\alpha)\equiv G(-x;-\alpha)$. A similar remark applies to Corollary~\ref{cor5}. 
\end{remark} 

Choosing $\cN=N+M$ ($N,M\in\N$) and 
\begin{equation} 
(m_J,X_J) = \begin{cases} (1,x_J) & \text{ for } J=1,\ldots,N \\ (1/g,y_{J-N}) & \text{ for } J-N=1,\ldots,M \end{cases} 
\end{equation} 
one obtains a common kernel function for the two pairs of analytical difference operators $(S^\pm_N(\vx;g,\beta),-S^\pm_M(\vy,1/g,g\beta)/g)$: 

\begin{corollary} 
\label{cor3} 
The function
\begin{equation} 
\label{tFNM}
\tilde F_{N,M}(\vx,\vy;g,\beta) = \Psi_N(\vx;g,\beta)\Psi_M(\vy;1/g,g\beta) \prod_{j=1}^N\prod_{k=1}^M s(x_j-y_k) 
\end{equation} 
obeys 
\begin{equation} 
\label{SI2} 
\left( S^\pm_N(\vx;g,\beta) + \frac1g S^\pm_M(\vy;1/g,g\beta) -\frac{s(\ii g\beta (N+M/g) )}{\ii g\beta s'(0) }\right)  \tilde F_{N,M}(\vx,\vy;g,\beta) =0
\end{equation} 
for arbitrary $N,M$ in the rational, trigonometric and hyperbolic cases, and for 
\begin{equation} 
N + M/g =0 \quad \mathrm{(IV)}
\end{equation} 
in the elliptic case. 
\end{corollary} 

Choosing $\cN=N+\tN$ ($N,\tN\in\N$) and \Ref{mX} one obtains an important common eigenstate of the deformed Ruijsenaars operators in \Ref{SNM}: 

\begin{corollary}
\label{cor4} 
The function
\begin{equation} 
\label{PsiNM} 
\Psi_{N,\tN}(\vx,\vtx;g,\beta) = \frac{\Psi_N(\vx;g,\beta)\Psi_{\tN}(\vtx;1/g,-g\beta)}{\prod_{j=1}^N\prod_{k=1}^{\tN} \left( s(x_j-\tx_k +\ii g\beta/2-\ii\beta/2)s(x_j-\tx_k -\ii g\beta/2+\ii\beta/2) \right)^{1/2} }
\end{equation} 
obeys 
\begin{equation} 
\left( S^\pm_{N,\tN}(\vx,\vtx;g,\beta) - \frac{s(\ii g\beta (N-\tN/g))}{\ii g\beta s'(0)) } \right) \Psi_{N,\tN}(\vx,\vtx;g,\beta) =0 
\end{equation} 
for arbitrary $N,\tN$ in the rational, trigonometric and hyperbolic cases, and for
\begin{equation} 
N-\tN/g=0\quad \mathrm{(IV)}
\end{equation} 
in the elliptic case. 
\end{corollary} 

Choosing $\cN=N+\tN+M+\tM$ ($N,\tN,M,\tM\in\N$) and 
\begin{equation} 
(m_J,X_J) = \begin{cases} (1,x_J) & \text{ for } J=1,\ldots,N \\  (-1/g,\tx_{J-N}) & \text{ for } J-N=1,\ldots,\tN \\ (-1,y_{J-N-\tN}) & \text{ for } J-N-\tN=1,\ldots,M \\  (1/g,\ty_{J-N-\tN-M}) & \text{ for } J-N-\tN-M=1,\ldots,\tM \\ 
\end{cases} 
\end{equation} 
one obtains a common kernel function for the two pairs of deformed analytical difference operators $(S^\pm_{N,\tN}(\vx,\vtx),S^\pm_{M,\tM}(-\vy,-\vty))$: 

\begin{corollary} 
\label{cor5} 
The function
\begin{equation} 
\label{FNMNM}
\begin{split} 
F_{N,\tN,M,\tM}(\vx,\vtx,\vy,\vty;g,\beta) = \Psi_{N,\tN}(\vx,\vtx;g,\beta)\Psi_{M,\tM}(-\vy,-\vty;g,\beta) \\ \times \left( \prod_{j=1}^N\prod_{k=1}^M \frac{G(x_j-y_k-\ii g\beta /2;\beta)}{G(x_j-y_k+\ii g\beta /2;\beta)}\right)\left( \prod_{j=1}^N\prod_{k'=1}^{\tM} s(x_j-\ty_{k'}) \right)\\ \times  
\left( \prod_{j'=1}^{\tN}\prod_{k=1}^{M} s(\tx_{j'}-y_k) \right) \left( \prod_{j'=1}^{\tN}\prod_{k'=1}^{\tM} \frac{G(\tx_{j'}-\ty_{k'}+\ii \beta /2;g\beta)}{G(\tx_{j'}-\ty_{k'}-\ii \beta /2;g\beta)}\right) 
\end{split} 
\end{equation} 
obeys 
\begin{equation} 
\begin{split} 
\left( S^\pm_{N,\tN}(\vx,\vtx;g,\beta) - S^\pm_{M,\tM}(-\vy,-\vty;g,\beta) -\frac{s(\ii g\beta (N-M -(\tN-\tM)/g))}{\ii g\beta s'(0)}
\right)\\ \times   F_{N,\tN,M,\tM}(\vx,\vtx,\vy,\vty;g,\beta) =0
\end{split} 
\end{equation} 
for arbitrary $N,\tN,M,\tM$ in the rational, trigonometric and hyperbolic cases, and for 
\begin{equation} 
\label{BC5}
N-M-(\tN-\tM)/g=0 \quad \mathrm{(IV)} 
\end{equation} 
in the elliptic case. 
\end{corollary} 

Corollary \ref{cor5} is the most general result and implies all other results stated in this section as special cases. 

\section{Comparison with previous results} 
\label{sec4} 
\subsection{Non-relativistic limit}
\label{sec4.1} 
To set our results in perspective we discuss the source identity and its implications for Calogero-Sutherland (CS) models. 

The $A_{N-1}$ CS system is defined by the Hamiltonian
\begin{equation} 
\label{HN} 
H_{N}(\vx;g) = -\sum_{j=1}^N \frac{\partial^2}{\partial x_j^2} + \sum_{1\leq j<k\leq N} 2g(g-1)V(x_j-x_k)
\end{equation} 
\begin{equation} 
V(x) = -\frac{\partial^2}{\partial x^2}\log s(x) 
\end{equation} 
for $s(x)$ in \Ref{s}. It can be obtained as the following non-relativistic limit of the analytic difference operators in \Ref{Spm},
\begin{equation} 
H_N(\vx;g) =\lim_{\beta\to 0}\frac{1}{\beta^2}\left( S^+(\vx;g,\beta)+S^-(\vx;g,\beta) - 2N  \right) +\frac13 N g^2\frac{s'''(0)}{s'(0)}. 
\end{equation} 

The corresponding source identity is as follows \cite{sourceAN}:  

\begin{theorem}[Non-relativistic source identity] 
\label{thmnrsource} 
Let
\begin{equation} 
\begin{split} 
\label{Sen} 
\cH_\cN(\vX;\vm) &= -\sum_{J=1}^{\cN}\frac1{m_J}\frac{\partial^2}{\partial X_J^2} + \sum_{1\leq J<k\leq N} \gamma_{JK} V(X_J-X_K) \\
\gamma_{JK}&=(m_J+m_K)g(m_Jm_K g-1)
\end{split} 
\end{equation} 
with $\vX=(X_1,\ldots,X_\cN)\in\C^\cN$, $\cN\in\N$, $g\in\C$, $\vm\in\C^{\cN}$, and 
\begin{equation} 
\Phi_{\mathrm{nr}}(\vX;\vm)= \prod_{1\leq J<K<\cN} s(X_J-X_K)^{m_Jm_K g}  . 
\end{equation} 
Then 
\begin{equation} 
\label{nrsource} 
\left( \cH_\cN(\vX;\vm) -\cE_\cN(\vm) \right)  \Phi_{\mathrm{nr}}(\vX;g) = 0 
\end{equation} 
with a known constant $\cE_\cN(\vm)$ \cite{sourceAN} holds true for arbitrary $\vm$ in the rational, trigonometric and hyperbolic cases; in the elliptic case \Ref{nrsource} holds true if and only if $\vm$ satisfies the balancing condition in \Ref{BC}. 
\end{theorem} 

Similarly as in Section~\ref{sec3} one can obtain from this kernel function identities. We only mention the kernel function for the pair of CS Hamiltonians $(H_N(\vx;g),H_M(\vy;g))$ given by 
\begin{equation}  
F_{N,M,\mathrm{nr} }(x,y) = \frac{\left(\prod_{1\leq j<k\leq N}s(x_j-x_k)^g\right)\left(\prod_{1\leq j<k\leq M}s(y_j-y_k)^g\right)}{\prod_{j=1}^N\prod_{k=1}^Ms(x_j-y_k)^g}, 
\end{equation} 
and the non-relativistic limit of \Ref{SNM}
\begin{equation} 
H_{N,\tN}(\vx,\vtx;g) = H_{N}(\vx;g) - g H_{\tN}(\vtx;1/g) + 2(1-g)\sum_{j=1}^N\sum_{k=1}^{M} V(x_j-\tx_k). 
\end{equation} 
defining a deformed CS system. 

Note that the balancing condition in the elliptic case is the same in the relativistic and non-relativistic cases. However, in the non-relativistic case, generalizations of these results to arbitrary $\vm$ (not satisfying the balancing condition) are known: the generalization of \Ref{nrsource} is 
\begin{equation} 
\label{nrsource1} 
\left( 4 g \left(\sum_{J=1}^{\cN} m_J \right) r \frac{\partial}{\partial a} +\cH(\vX;\vm) -\cE_\cN(\vm)  \right)  \Phi_{\mathrm{nr}}(\vX;g) = 0 \quad \mathrm{(IV)} 
\end{equation} 
with a known constant $\cE_\cN(\vm) $ \cite{sourceAN} (corresponding generalization for kernel function identities etc.\ \cite{ELrem} can be obtained from this in a simple way \cite{sourceAN}); note that the balancing condition in  \Ref{BC}  ensures that the $a$-derivative term is absent. It would be interesting to find a relativistic generalization of this generalized elliptic source identity. 

\subsection{Source identities for Macdonald-Rujsenaars-type operators} 
\label{sec4.2} 
Corollary~\ref{cor2} for $N=M$ is due to Ruijsenaars \cite{RK}; see also \cite{HR} for generalizations to $N\neq M$ in the hyperbolic case.  The kernel function identities in Corollaries~\ref{cor2} and \ref{cor3} where previously obtained in \cite{KNS}, Theorems 2.1 and 2.2; to see this, note that Lemma~\ref{LemmaA} implies 
\begin{equation} 
\label{cA} 
\begin{split} 
\cA^\pm_{N}(\vx; g,\beta) \equiv &\frac{\ii g\beta s'(0)}{s(\ii g\beta)}\frac{1}{\Psi_{N}(\vx;g,\beta) } S^\pm_{N}(\vx;g,\beta)\Psi_{N}(\vx;g,\beta) \\  = & \sum_{j=1}^N \left( \prod_{j'\neq j}\frac{s(x_j-x_{j'} \mp \ii g\beta)}{s(x_j-x_{j'})}\right)  \ee{\mp\ii\beta\frac{\partial}{\partial x_j} } , 
\end{split} 
\end{equation} 
and \Ref{SI1} and \Ref{SI2} are equivalent to
\begin{equation} 
\left( \cA^\pm_N(\vx;g,\beta) - \cA^\pm_M(\vy;g,\beta) - \frac{s(\ii g\beta(N-M) )}{s(\ii g\beta)} \right) \prod_{j=1}^N\prod_{k=1}^M \frac{G(x_j+v+y_k-\ii g\beta /2;\beta)}{G(x_j+v+y_k+\ii g\beta /2;\beta)} =0 
\end{equation} 
(we replaced $(x_j,y_k)$ by $(x_j+v,-y_k)$; see Remark~\ref{remv} for details) and 
\begin{equation} 
\left( s(\ii g\beta) \cA^\pm_N(\vx;g,\beta) + s(\ii\beta)\cA^\pm_M(\vy;1/g,g\beta) -s(\ii g\beta(N +M/g)) \right)  \prod_{j=1}^N\prod_{k=1}^M s(x_j-y_k+v )  =0 , 
\end{equation} 
respectively (in \cite{KNS} only the upper identities are stated). 

We also note the following generalization of \Ref{cA}: Corollary~\ref{cor5} implies that the function 
\begin{equation} 
\begin{split} 
\left( \prod_{j=1}^N\prod_{k=1}^M \frac{G(x_j+y_k+v -\ii g\beta /2;\beta)}{G(x_j+y_k+v+\ii g\beta /2;\beta)}\right)\left( \prod_{j=1}^N\prod_{k'=1}^{\tM} s(x_j+\ty_{k'}+v) \right)\\ \times  
\left( \prod_{j'=1}^{\tN}\prod_{k=1}^{M} s(\tx_{j'}+y_k+v) \right) \left( \prod_{j'=1}^{\tN}\prod_{k'=1}^{\tM} \frac{G(\tx_{j'}+\ty_{k'}+v+\ii \beta /2;g\beta)}{G(\tx_{j'}+\ty_{k'}+v-\ii \beta /2;g\beta)}\right) 
\end{split} 
\end{equation} 
is a common kernel function for the pairs of operators $(\cA^\pm_{N,\tN}(\vx,\vtx;g,\beta),\cA^\pm_{M,\tM}(\vy,\vty;g,\beta))$ with 
\begin{equation} 
\label{cApm} 
\begin{split} 
&\cA^\pm_{N,\tN}(\vx,\vtx; g,\beta) \equiv \frac{\ii g\beta s'(0)}{s(\ii g\beta)}\frac{1}{\Psi_{N,\tN}(\vx,\vtx;g,\beta) } S^\pm_{N,\tN}(\vx,\vtx;g,\beta)\Psi_{N,\tN}(\vx,\vtx;g,\beta) \\ & = 
 \sum_{j=1}^N \left(\prod_{j'\neq j}  \frac{s(x_j-x_{j'} \mp \ii g\beta)}{s(x_j-x_{j'})} \right) \left( \prod_{k=1}^{\tN}\frac{s(x_j-\tx_k\mp \ii g\beta/2\pm  \ii\beta/2)}{s(x_j-\tx_k \mp \ii g\beta/2 \mp \ii\beta/2)} \right) \ee{\mp\ii\beta\frac{\partial}{\partial x_j} } 
 \\ & - \sum_{k=1}^{\tN} \frac{s(\ii \beta)}{s(\ii g\beta)} \left( \prod_{k'\neq k} \frac{s(\tx_k-\tx_{k'} \pm \ii \beta)}{s(\tx_k-\tx_{k'})} \right) \left( \prod_{j=1}^{N}\frac{s(\tx_k-x_j \pm \ii\beta/2\mp \ii g\beta/2)}{s(\tx_k-x_j \pm \ii\beta/2\pm \ii g\beta/2)} \right)\ee{\pm\ii g \beta\frac{\partial}{\partial \tx_k} } 
\end{split} 
\end{equation} 
(the latter equality is obtained from Lemma~\ref{LemmaA} in the special case \Ref{mX}; the constant appearing in the kernel function identity is $c=s(\ii g\beta(N-M-(\tN-\tM)/g))/s(\ii g\beta)$, and in the elliptic case there is the balancing condition in \Ref{BC5}). The operators $\cA^+_{N,\tN}(\vx,\vtx; g,\beta)$ in the trigonometric case agree, up to additive and multiplicative constants, with the deformed Macdonald-Ruijsenaars operator studied by Sergeev and Veselov \cite{SV2}; to see this, change variables 
 \begin{equation} 
q=\ee{2r\beta},\quad t=\ee{-2rg\beta}, \quad z_j=\ee{2\ii r x_j}/t^{1/2},\quad \tilde z_k=\ee{2\ii r\tx_k}/q^{1/2}
\end{equation} 
to obtain $\cA^+_{N,\tN}=t^{-(N-1)/2}q^{-\tN/2}\cM_{N,\tN}$ with 
\begin{equation} 
\begin{split} 
\label{SV2}
\cM_{N,\tN}(\vz,\vtz) = & \Biggl[ \sum_{j=1}\left(\prod_{j'\neq j}\frac{z_j-tz_{j'}}{z_j-z_{j'}} \right)\left(\prod_{k=1}^{\tN}\frac{z_j-q\tz_k}{z_j-\tz_k} \right)T_{q,z_j}  \\  & + 
 \frac{1-q}{1-t} \sum_{k=1}\left(\prod_{k'\neq k}\frac{\tz_k-q\tz_{k'}}{\tz_k-\tz_{k'}} \right)\left(\prod_{j=1}^{N}\frac{\tz_k-tz_j}{\tz_k-z_j} \right)T_{t,\tz_k} \Biggr]
 \end{split}
 \end{equation} 
and $T_{q,z}$ such that $(T_{q,z} f)(z)\equiv f(qz)$ for meromorphic functions $f(z)$ etc., which is equal up to a constant  to the operator in \cite{SV2}, Eq.\ (1) (note that replacing $T_{q,z_j}$ and $T_{t,\tz_k}$ in \Ref{SV2} by  $(T_{q,z_j}-1)$ and $(T_{t,\tz_k}-1)$ amounts to adding an additive constant due to Lemma~\ref{Lemma1}). In \cite{SV2} a kernel function for the pair of operators $(\cM_{N,\tN}(\vz,\vtz),\cM_{M,\tM}(\vw,\vtw))$ for $\tM=0$ is given; our results extend this to $\tM\neq 0$. Moreover, our results also imply that the very same kernel function is also a kernel function for the pair of operators $(\cM^-_{N,\tN}(\vz,\vtz),\cM^-_{M,\tM}(\vw,\vtw))$ with 
\begin{equation} 
\begin{split} 
\cM^-_{N,\tN} = & \Biggl[ \sum_{j=1}\left(\prod_{j'\neq j}\frac{tz_j-z_{j'}}{tz_j-tz_{j'}} \right)\left(\prod_{k=1}^{\tN}\frac{tz_j-\tz_k}{tz_j-q\tz_k} \right)T_{1/q,z_j}  \\  & + 
 \frac{t(1-q)}{q(1-t)} \sum_{k=1}\left(\prod_{k'\neq k}\frac{q\tz_k-\tz_{k'}}{q\tz_k-q\tz_{k'}} \right)\left(\prod_{j=1}^{N}\frac{q\tz_k-z_j}{q\tz_k-tz_j} \right)T_{1/t,\tz_k} \Biggr] . 
 \end{split} 
 \end{equation}

\section{Final remarks} 
\label{sec6} 
A key property which makes a source identity useful to derive kernel function identities is that pairs of variables corresponding to mass parameters $(m_J,m_K)$ with $m_K=-m_J$ and $m_K=1/gm_J$ decouple. In the non-relativistic case, this is a consequence of the formula for the coupling constants $\gamma_{JK}$ in \Ref{Sen}, which implies $\gamma_{JK}=0$ in these cases. However, this property seems to be "put in by hand" in the relativistic case (by setting $f_\pm(x;m_J,m_K)=1$ in these cases; see \Ref{fpm}). It thus is interesting to note that one can write the operators $\cS^\pm_{\cN}(\vX;\vm)$ in \Ref{cSpm} in a way that brings them closer in this regard to the non-relativistic counterpart in \Ref{Sen}: as is seen by a simple computations, 
\begin{equation} 
\label{cSpm1}
\begin{split} 
\cS^\pm_{\cN}(\vX;\vm)  = &\sum_{J=1}^{\cN}\frac{s(\ii g \beta m_J)}{\ii g\beta s'(0)} \left(\prod_{K\neq J} \left(\frac{s(X_J-X_K \mp\ii \xi(m_J,m_K) \mp\ii g\beta m_K)}{s(X_J-X_K\mp\ii \xi(m_J,m_K))}\right)^{1/2} \right)
\\  & \times \exp\left(\mp\ii\frac{\beta}{m_J}\frac{\partial}{\partial x_J}\right) \left( \prod_{K\neq J}\left(\frac{s(X_J-X_K \pm \ii \xi(m_J,m_K) \pm\ii g\beta m_K)}{s(X_J-X_K\pm\ii \xi(m_J,m_K))}\right)^{1/2}\right) 
\end{split} 
\end{equation} 
with $\xi(m,m')$ in \Ref{xi}. Note that, if we write the operators in this form, they have the following property which is well-known in the standard case \Ref{Spm}: the similarity transformation by the "groundstate" $\Phi(\vX;\vm)$ is equivalent to removing the square roots and the factors after the shift operator $ \exp\left(\mp\ii(\beta/m_J)\partial/\partial x_J\right) $; see Lemma~\ref{LemmaA}.  

We conclude with a list of research problems motivated by our results (some were already mentioned). 

\begin{itemize} 
\item Find a relativistic source identity for the elliptic case but without balancing condition (relativistic generalization of \Ref{nrsource1}; we mention interesting ideas of E. Rains, M. Noumi and S. Razamat on this presented in their talks at the  workshop "Elliptic Integrable Systems and Hypergeometric Functions" in July 2013 at the Lorentz Center in Leiden\footnote{Available on "http://www.lorentzcenter.nl/lc/web/2013/541/info.php3?wsid=541\&venue=Oort".}). 

\item Prove quantum integrability of the deformed Ruijsenaars operators in \Ref{SNM} in the elliptic case (generalization of results in \cite{R} and \cite{SV2}).

\item Find a source identity, and thus all kernel functions, for the relativistic generalization of BC$_N$-type Calogero-Sutherland models model due to van Diejen \cite{vD} (generalization of results in \cite{sourceBCN}). 

\item Construct eigenfunctions of the (deformed) Ruijsenaars operators in the trigonometric and elliptic cases using  the kernel functions in Corollaries \ref{cor2} and \ref{cor5}  (generalization of results in \cite{EL0,eCS}). 
\end{itemize} 

\bigskip

\noindent {\bf Acknowledgments:} We thank Jan Felipe van Diejen and Simon Ruijsenaars  for valuable discussions during a mini workshop at KTH in August 2012 when this work was started. We are grateful to the organizers of the workshop "Elliptic Integrable Systems and Hypergeometric Functions" (July 2013, Lorentz Center), and we acknowledge interesting discussions with Masatoshi Noumi, Michael Pawellek, Michael Schlosser, Junichi Shiraishi, Jasper Stokman. This work was partly supported by the G\"oran Gustafsson Foundation and the Swedish Research Council (VR) under  contract no. 621-2010-3708. The work of M. H. was partially supported by the EPSRC (grant EP/K010123/1).

\appendix
\section{Gamma functions}
\label{appA} 
For the convenience of the reader we collect formulas for Gamma functions defined by the relation in \Ref{G1}; see e.g.\ \cite{RGamma,KNS} for more detailed discussions. 

It is interesting to note that, for integer $n\neq 0$, one has 
\begin{equation} 
\frac{G(x+\ii n\alpha;\alpha)}{G(x-\ii n\alpha;\alpha)} = \begin{cases} \prod_{k=0}^{2n-1} s(x +\ii (n-k-1/2)\alpha) \text{ for } n>0 \\  \prod_{k=0}^{-2n-1} s(x +\ii(-n-k-1/2)\alpha)^{-1} \text{ for } n<0 
\end{cases}
\end{equation} 
(this following from \Ref{G1} by iteration; we set $c=1$ without loss of generality). We recall that these cases are known to be of special significance: it was shown by Etingof and Styrkas \cite{ES} that the (trigonometric) Macdonald-Ruijsenaars operators are algebraically integrable and admit eigenfunctions of so-called Baker-Akhiezer type; see also \cite{OAC}. Thus, if one restricts Theorem~\ref{Theorem1} to $m_J\in\{ 1,-1\}$  and integer $g$, the Gamma function is not needed. In particular, the ground states and kernel functions for the "standard" case (i.e.\ Corollaries \ref{cor1} and \ref{cor2})  and integer $g$ can be constructed without Gamma functions. However, in general, the Gamma function is required.  

We also note that, if $G_1(x;\alpha)$ is a function satisfying \Ref{G1}, then also 
\begin{equation} 
\label{Gj} 
G_2(x;\alpha)\equiv G_1(-x;-\alpha),\quad G_3(x;\alpha)\equiv 1/G_1(x;-\alpha),\quad G_4(x;\alpha)\equiv 1/G_1(-x;\alpha) 
\end{equation} 
are such functions (this is easy to check; note that the constant $c$ is not always the same).  Depending on the case, not all of these four Gamma functions are different; see below. 

Below we give a natural solution $G_1$ of \Ref{G1} for each of the four cases. As mentioned in the main text, we find it convenient to take $G=G_1$ for $\Re(\alpha)>0$ and $G=G_2$ for $\Re(\alpha)<0$, but this is not necessary. In fact, in all but the rational cases, the natural such extension of $G_1$ is $G_3$, and $G_3$ is different from $G_2$. To explain what is meant by "natural extension", and as a preparation for the following, we find it useful to discuss in the following paragraph a function which is the building block for the trigonometric- and elliptic Gamma functions. 

Consider the function 
\begin{equation} 
\label{f} 
f(z;q) \equiv \prod_{k=1}^\infty \frac1{(1-q^{2k-1}z)}   \quad (|q|<1)\\
\end{equation} 
of two complex variables $z$ and $q$ (here and in the following we indicate the restriction of variables for a given formula to be well-defined, for example, the restriction $|q|<1$ in \Ref{f} indicates that this defines a function $f(z;q)$ for  arbitrary complex $z$ and complex $q$ such that $|q|<1$). This function has the following natural extension
\begin{equation} 
\label{f2} 
f(z;q) \equiv \prod_{k=1}^\infty (1-(1/q)^{2k-1}z)   \quad (|q|>1), 
\end{equation} 
in the sense that both functions in \Ref{f} and \Ref{f2} are analytical continuations of the following one 
\begin{equation} 
\label{f3} 
f(z;q) = \exp\left(  \sum_{n=1}^\infty \frac{z^n}{n(q^{-n}-q^n) }\right) \quad (|z|<\max(|q|, |1/q|), |q|\neq 1)
\end{equation} 
(to see this,  write in both cases \Ref{f} and \Ref{f2} the product as exponential of a series, insert the Taylor series of the log, exchange summations,  and sum up a geometric series). Thus, in this sense, the function defined in \Ref{f} for $|q|<1$ has a natural extension to all $|q|\neq 1$ such that $f(z;q)=1/f(z;1/q)$. A second noteworthy property of this function is the functional equation 
\begin{equation} 
\label{f1} 
f(q z;q)=(1-z)f(z/q;q)\quad (|q|\neq 1) 
\end{equation} 
which is easily checked by computations.  Note that, for fixed $z$ such that $|z|<1$, the function $f(z;q)$ has essential singularities in $q$ which lie dense on the unit circle $|q|=1$. 

\subsection{Rational case} 
In this case, $G_1(x;\alpha) \equiv  \Gamma(1/2+x/\ii\alpha)$ satisfies \Ref{G1} with $c=1/\ii\alpha$, as already mentioned in the main text. Note that this function, for arbitrary $\alpha\in\C\setminus\{0\}$,  is a  meromorphic function in $x\in\C$ which is non-zero everywhere and has simple poles at $x=- \ii\alpha(n-1/2)$, $n\in\N$. Obviously, $G_2=G_1$ and $G_4=G_3$.  Moreover, 
\begin{equation} 
G_3(x;\alpha) = G_1(x;\alpha) \frac{\cosh(\pi x/\alpha)}{\pi}
\end{equation} 
(the latter follows from the well-known functional identity $\Gamma(x)\Gamma(1-x)=\pi/\sin(\pi x)$). 

\subsection{Trigonometric case} 
A standard trigonometric gamma function is given by \cite{RGamma}
\begin{equation} 
\label{GRtrig} 
\GR(r,\alpha;x) = \prod_{k=1}^\infty(1-\ee{-r\alpha (2k-1)}\ee{2\ii r x})^{-1}   \quad (\Re(\alpha)>0) . 
\end{equation} 
equal to the function $f(z;q)$ in \Ref{f} with $z=\exp(2\ii rx)$ and $q=\exp(-r\alpha)$. Thus the results discussed after \Ref{f} above imply that
\begin{equation} 
\label{Gtrig} 
G_1(x;\alpha) \equiv  \ee{-r x^2/2\alpha}\GR(r,\alpha;x) \quad (\Re(\alpha)>0)  
\end{equation} 
has a natural extension to $\Re(\alpha)\neq 0$ such that 
\begin{equation} 
\label{G1G3} 
G_1(x;\alpha) = 1/G_1(x;-\alpha)\quad (\Re(\alpha)\neq 0), 
\end{equation} 
and which satisfies \Ref{G1} with $c=-2\ii r$. Thus we have the following relations between the functions in \Ref{Gj}, $G_3=G_1$ and $G_4=G_2$.  However, $G_2$ is different from $G_1$. 

We also note that following representation of this function \cite{RGamma},  
\begin{equation} 
G_1(x;\alpha) = \exp\left( -\frac{rx^2}{2\alpha} + \sum_{n=1}^\infty \frac{\exp(2\ii n rx)}{2n\sinh(nr\alpha)} \right) \quad  
(\Im(x)>-|\Re(\alpha)|/2). 
\end{equation}

\subsection{Hyperbolic case} 
We note that it is not possible to obtain the hyperbolic Gamma function from the trigonometric one by analytical continuation $r\to \ii \pi/a$, as one might naively expect. 

A standard hyperbolic Gamma function is given by \cite{RGamma} 
\begin{equation} 
\begin{split}
\GR(a,\alpha;x) = \exp\left( \int_0^\infty \frac{dy}{y}\left(\frac{\sin(2xy)}{2\sinh(ay)\sinh(\alpha y)}-\frac{x}{a\alpha y}\right)\right)\\(|\Im(x)|<\Re(a+\alpha)/2,\quad \Re(\alpha)>0)
\end{split}
\end{equation} 
which has  a natural extension to all $\Re(\alpha)\neq 0$ such that 
\begin{equation} 
\label{dual} 
\GR(a,\alpha;x) = 1/\GR(a,\alpha;-x) =1/\GR(a,-\alpha;x) 
\end{equation} 
(this is suggested by the definition). Moreover, this function satisfies
\begin{equation} 
\frac{\GR(a,\alpha;x+\ii\alpha/2)}{\GR(a,\alpha;x-\ii\alpha/2)} = 2\,\cosh(\pi x/a) . 
\end{equation} 
Thus 
\begin{equation} 
G_1(x;\alpha) = \GR(a,\alpha;x+\ii a/2) \quad (|\Im(x)+a/2|<\Re(a+\alpha)/2)
\end{equation} 
satisfies \Ref{G1} with $c=2\pi\ii/a$, and the relations between the functions in \Ref{Gj}  are as in the trigonometric case. 

\subsection{Elliptic case} 
We note that, in the elliptic case, the function in \Ref{s} can be written as
\begin{equation} 
\label{sell} 
s(x) = \frac1r\sin(rx)\prod_{\ell =1}^\infty \frac{(1-\ee{-2\ell r a}\ee{2\ii rx})(1-\ee{-2\ell r a}\ee{-2\ii r x})}{(1-\ee{-2\ell ra})^2}  
\end{equation} 
(see e.g.\ \cite{WW}, Section 20.421). 

A standard elliptic Gamma function is given by \cite{RGamma}
\begin{equation}
\label{GRell} 
\GR(r,a,\alpha;x)\equiv \prod_{k,\ell=1}^\infty \frac{(1-\ee{-(2k-1)r\alpha}\ee{-(2\ell-1)ra}\ee{-2\ii rx})}{(1-\ee{-(2k-1)r\alpha}\ee{-(2\ell-1)ra}\ee{2\ii rx})}\quad (\Re(\alpha)>0) 
\end{equation} 
and satisfies the $\GR(r,a,\alpha;x) = 1/\GR(r,a,\alpha;-x)$ (this is implied by the definition). The r.h.s.\ in \Ref{GRell} is obviously a product of functions $f(z;q)$ in \Ref{f} and \Ref{f1} with $q=\exp(- r\alpha)$, $z=\exp(-(2\ell -1)ra\mp 2\ii rx)$. Thus the results discussed after \Ref{f} above imply that
\begin{equation} 
\label{G1ell} 
G_1(x;\alpha) \equiv  \ee{-r x^2/2\alpha}\GR(r,a,\alpha;x-\ii a/2) \quad (\Re(\alpha)>0)  
\end{equation} 
has an natural extension to $\Re(\alpha)\neq 0$ such that \Ref{G1G3} holds true, and which satisfies \Ref{G1} for $c=-2\ii r\prod_{\ell=1}^\infty(1-\ee{-2\ell ra})^2$:  
\begin{equation} 
\frac{G_1(x+\ii\alpha/2;\alpha)}{G_1(x-\ii\alpha/2;\alpha)} = \ee{-\ii rx} \prod_{\ell=1}^\infty (1-\ee{-2\ell ra}\ee{-2\ii rx}) (1-\ee{-(2\ell-2)ra}\ee{2\ii rx}) = c\,  s(x) \quad (\Re(\alpha)>0). 
\end{equation} 
The reciprocity relation and \Ref{G1G3} imply 
\begin{equation} 
\label{G2ell} 
G_1(-x;-\alpha) =  \ee{r x^2/2\alpha}\GR(r,a,\alpha;x+\ii a/2) \quad (\Re(\alpha)\neq 0) . 
\end{equation} 
Thus the relations between the functions in \Ref{Gj} are as in the trigonometric case.

We also note \cite{RGamma} 
\begin{equation}
G_1(x;\alpha) = \exp\left(-\frac{rx^2}{2\alpha} +  \sum_{n=1}^\infty \frac{\sinh(n r a+2\ii nrx)}{2n\sinh(nr\alpha)\sinh(nra)}\right) \quad  |\Im(x)-a/2|<|\Re(a+\alpha)|/2 .
\end{equation} 

The trigonometric- and hyperbolic Gamma functions can be obtained from the elliptic one by suitable limits \cite{RGamma}. 

\section{Proof of Lemma~\ref{Lemma1}}
\label{appB} 
This can be proved by a classic argument \cite{R}: regard the l.h.s.\ of \Ref{WH} as meromorphic function of one variable $Z_1$ (say) by fixing $\gamma$ and the other variable $Z_{J\neq 1}$ so that $Z_J\neq Z_K$ for all $J\neq K$, $J,K\neq 1$. It is obvious that the only singularities of this function are poles of first order determined by the zeros of the function $s(x)$. Computing the residues at these poles one finds that they all vanish, and thus Liouville's theorem \cite{WW} implies that this function is a constant in $Z_1$ (say). By symmetry, this implies that the l.h.s.\ of \Ref{WH} is a constant in all $Z_J$. One then computes this constant by taking suitable limits, e.g. 

As a representative example,  we compute the residue of the l.h.s.\ of \Ref{WH} at $Z_1=Z_2$: 
\begin{equation} 
\label{Res} 
\begin{split} 
\mathrm{Res}_{Z_1=Z_2}\equiv  \lim_{Z_1\to Z_2}(Z_1-Z_2) \times \left( \text{l.h.s.\ of \Ref{WH}}\right) \\ 
 = \lim_{Z_1\to Z_2} \Biggl( s(\gamma m_1)\frac{(Z_1-Z_2)s(Z_1-Z_2+\gamma m_2)}{s(Z_1-Z_2)}\prod_{K=3}^{\cN}\frac{s(Z_1-Z_K+\gamma m_K)}{s(Z_1-Z_K)}  \\  - s(\gamma m_2)\frac{(Z_2-Z_1)s(X_2-X_1+\gamma m_1)}{s(Z_2-Z_1)}\prod_{K=3}^{\cN}\frac{s(Z_2-Z_K+\gamma m_K)}{s(Z_2-Z_K)}\Biggr) \\  =  \left( s(\gamma m_1)\frac{s(\gamma m_2)}{s'(0)} - s(\gamma m_2)\frac{s(\gamma m_1)}{s'(0)}\right) \prod_{K=3}^{\cN}\frac{s(Z_2-Z_K+\gamma m_K)}{s(Z_2-Z_K)} =0. 
\end{split} 
\end{equation} 
Obviously, by symmetry, we can conclude from this that $\mathrm{Res}_{Z_J=Z_K}=0$ for all $J\neq K$. 

The different cases differ by the number of poles to be checked and by how the constant is computed; we thus discuss them separately. 

\subsection{Rational case} 
For $s(x)=x$ the only poles of the l.h.s.\ in \Ref{WH} are at $Z_J=Z_K$, $J\neq K$. We already showed that the residues at these poles all are zero. The l.h.s.\ of \Ref{WH} thus is bounded and analytic in all variables and thus has an analytic continuation which is a constant. To compute this constant we take the limit $Z_1\to\infty$ keeping all $Z_{J>1}$ constant, $Z_2\to\infty$ keeping all $Z_{J>2}$ constant, etc., to obtain $\sum_{J} \gamma m_K$ equal to the r.h.s.\ in \Ref{WH}. 

\subsection{Trigonometric case}
For $s(x)=(1/r)\sin(rx)$ the poles of the l.h.s.\ in \Ref{WH} are at 
\begin{equation} 
Z_J=Z_K +n\pi/r,\quad n\in\Z. 
\end{equation} 
Since $s(x)$ is $\pi/r$-antiperiodic, the l.h.s.\ of \Ref{WH} is a $\pi/r$-periodic function in each variable $Z_J$ separately. Thus the residues at $Z_J=Z_K +n\pi/r$ are $n$-independent and, according to the computation above, all vanish. As above one concludes that the l.h.s.\ in \Ref{WH} is a constant in all variables $Z_J$. This constant can be computed by taking the $Z_J$ to infinity in the negative imaginary directions (such that $\lim_{Z_J\to-\ii\infty}s(Z_J+c)/s(Z_J)=\exp(\ii r c)$ for complex constants $c$), as above:  
\begin{equation} 
\begin{split} 
\lim_{Z_1\to-\ii \infty}\left( \text{l.h.s.\ of \Ref{WH}}\right) = s(\gamma m_1)\ee{ \ii r\sum_{J=2}^{\cN} \gamma m_J} \\  + \ee{-\ii r \gamma m_1}\sum_{J=2}^{\cN}s(\gamma m_J)\prod_{2\leq K\leq\cN; K\neq J}\frac{s(X_J-X_K+\gamma m_K)}{s(X_J-X_K)} , 
\end{split} 
\end{equation} 
\begin{equation} 
\begin{split} 
\lim_{Z_2,Z_1\to-\ii \infty}\left( \text{l.h.s.\ of \Ref{WH}}\right) = s(\gamma m_1)\ee{ \ii r\sum_{J=2}^{\cN} \gamma m_J} + \ee{-\ii r \gamma m_1}s(\gamma m_2)\ee{ \ii r\sum_{J=3}^{\cN} \gamma m_J}  \\  + \ee{-\ii r \gamma (m_1+m_2)}\sum_{J=3}^{\cN}s( \gamma m_J)\prod_{3\leq K\leq\cN; K\neq J}\frac{s(X_J-X_K+\gamma m_K)}{s(X_J-X_K)} , 
\end{split} 
\end{equation} 
etc., until 
\begin{equation} 
\begin{split} 
\lim_{Z_{\cN},\ldots,Z_2,Z_1\to-\ii \infty}\left( \text{l.h.s.\ of \Ref{WH}}\right) = s(\gamma m_1)\ee{ \ii r\sum_{J=2}^{\cN} \gamma m_J} + \ee{-\ii r \gamma m_1}s(\gamma m_2)\ee{ \ii r\sum_{J=3}^{\cN}\gamma m_J}  \\  +  \ee{-\ii r\gamma  (m_1+m_2)}s(\gamma  m_3)\ee{ \ii r\sum_{J=4}^{\cN} \gamma m_J} + \cdots + \ee{-\ii r \sum_{J=1}^{\cN-1}\gamma m_J} s(\gamma m_{\cN} \gamma)  . 
\end{split} 
\end{equation} 
Inserting $s(\gamma m_J)=(1/2r\ii)(\ee{\ii \gamma m_J} -\ee{-\ii \gamma m_J})$ one sees that the latter is a telescoping sum which can be simplified to   
\begin{equation} 
\begin{split} 
\lim_{Z_{\cN},\ldots,Z_2,Z_1\to-\ii \infty}\left( \text{l.h.s.\ of \Ref{WH}}\right) = \frac1{2r\ii}\left( \ee{ \ii r\sum_{J=1}^{\cN} \gamma m_J} - \ee{ -\ii r\sum_{J=1}^{\cN} \gamma m_J}\right) 
\end{split} 
\end{equation} 
equal to the r.h.s.\ in \Ref{WH}. 

\subsection{Hyperbolic case}
For $s(x)=(a/\pi)\sinh(\pi x/a)$ the result can be obtained by analytical continuation $r\to \ii \pi/a$ from the result in the trigonometric case. 

\subsection{Elliptic case} 
For $\sigma(x;\pi/2r,\ii a/2)\exp(-\eta_1 r x^2/\pi)$ the poles of the l.h.s.\ in \Ref{WH} are at
\begin{equation} 
Z_J=Z_K + n\pi/r + n' \ii a,\quad n,n'\in\Z. 
\end{equation} 
The l.h.s\ of \Ref{WH} is $\pi/r$-periodic, but it is only quasi-periodic with respect to $\ii a$-shifts: since 
\begin{equation} 
\frac{s(x_1 + n\pi/r + n' \ii a)}{s(x_2 + n\pi/r+n'\ii a)} = \frac{s(x_1)}{s(x_2)}\ee{-2\ii r n' (x_1-x_2)}
\end{equation} 
(see e.g.\ \cite{WW}) one gets
\begin{equation} 
\label{shift} 
\begin{split} 
\left( \text{l.h.s.\ of \Ref{WH}}\right)_{Z_1\to Z_1 +n\pi/r+n'\ii a} = s(\gamma m_1)\prod_{K=2}^{\cN} \frac{s(Z_1-Z_K+\gamma m_K)}{s(Z_1-Z_K)}\ee{-2\ii r n'\gamma\sum_{K=2}^{\cN} m_K} +\\ 
\sum_{J=2}^{\cN} s(\gamma m_J)\prod_{K\neq J} \frac{s(Z_J-Z_K+\gamma m_K)}{s(Z_J-Z_K)}\ee{2\ii r n'\gamma m_1} , 
\end{split} 
\end{equation} 
and thus the residue computation in \Ref{Res} has the following important twist for $n'\neq 0$: 
\begin{equation} 
\label{Res1} 
\begin{split} 
\mathrm{Res}_{Z_1=Z_2+n\pi/r+n'\ii a}\equiv  \lim_{Z_1\to Z_2}(Z_1-Z_2) \times \left( \text{l.h.s.\ of \Ref{WH}}\right)_{Z_1\to Z_1 +n\pi/r+n'\ii a} \\ 
 = \lim_{Z_1\to Z_2} \Biggl( s(\gamma m_1)\frac{(Z_1-Z_2)s(Z_1-Z_2+\gamma m_2)}{s(Z_1-Z_2)}\prod_{K=3}^{\cN}\frac{s(Z_1-Z_K+\gamma m_K)}{s(Z_1-Z_K)}\ee{-2\ii r n'\gamma\sum_{K=2}^{\cN} m_K}   \\  - s(\gamma m_2)\frac{(Z_2-Z_1)s(X_2-X_1+\gamma m_1)}{s(Z_2-Z_1)}\prod_{K=3}^{\cN}\frac{s(Z_2-Z_K+\gamma m_K)}{s(Z_2-Z_K)}\ee{2\ii r n'\gamma m_1}\Biggr) \\  =  \ee{2\ii r n'\gamma m_1} \left( s(\gamma m_1)\frac{s(\gamma m_2)}{s'(0)}\ee{-2\ii r n'\gamma\sum_{K=1}^{\cN} m_K}  - s(\gamma m_2)\frac{s(\gamma m_1)}{s'(0)}\right)\prod_{K=3}^{\cN}\frac{s(Z_2-Z_K+\gamma m_K)}{s(Z_2-Z_K)} . 
\end{split} 
\end{equation} 
We thus see that all residues $\mathrm{Res}_{Z_1=Z_2+n\pi/r+n'\ii a}$ vanish if and only if
\begin{equation} 
\label{BC1} 
\exp\left(-2\ii r \gamma\sum_{K=1}^{\cN} m_K \right) =1 . 
\end{equation} 
By symmetry, this condition is necessary and sufficient for all residues $\mathrm{Res}_{Z_J=Z_K+n\pi/r+n'\ii a}$, $K\neq J$ and $n,n'\in\Z$, to vanish. Since we assume $\Im(\gamma)\neq 0$ and $m_J\in\R$, we obtain the necessary and sufficient condition in \Ref{BC}. We conclude that, if \Ref{BC} holds true, then the l.h.s.\ of \Ref{WH} is a constant. To compute this constant we observe that \Ref{shift} implies that, if \Ref{BC1} holds true, then the l.h.s.\ of \Ref{WH} changes by a factor $\exp(2\ii r n'\gamma m_1)$ under a shift $Z_1\to Z_1 +n\pi/r+n'\ii a$ (say), which is only possible if this constant is zero.

\end{document}